\documentclass[journal]{IEEEtran}

\usepackage{float}
\floatstyle{ruled}
\newfloat{model}{thp}{lop}
\floatname{model}{Model}

\usepackage{multicol}

\usepackage{xcolor}

\newcommand{\amdq}{Dell PowerEdge R415 servers with Dual 2.8GHz AMD 6-Core Opteron 4184 CPUs and 64GB of memory}

\usepackage{xspace}

\usepackage{url}
\usepackage[cmex10]{amsmath}
\usepackage{amssymb,amsthm,bm,mathtools,mathabx}
\usepackage[pdftex]{graphicx}
\usepackage{authblk}

\usepackage{hyperref} 
\hypersetup{
breaklinks=true,
urlcolor= blue, 
linkcolor= blue,
citecolor=red, 
}

\newtheorem{theorem}{Theorem}[section]
\newtheorem{lemma}[theorem]{Lemma}

\newtheorem{corollary}[theorem]{Corollary}

\begin{document}

\title{The QC Relaxation: A Theoretical and Computational Study on Optimal Power Flow} 

\author{Carleton~Coffrin,
        Hassan~L.~Hijazi,
        and~Pascal~Van~Hentenryck
\thanks{All authors are members of the Optimisation Research Group, NICTA,
ACT 2601 Australia; and affiliated with the College of Engineering and Computer Science, Australian National University, ACT 0200, Australia.}}
\maketitle

\begin{abstract}
Convex relaxations of the power flow equations and, in
  particular, the Semi-Definite Programming (SDP) and Second-Order
  Cone (SOC) relaxations, have attracted significant interest in
  recent years. The Quadratic Convex (QC) relaxation is a departure
  from these relaxations in the sense that it imposes constraints to
  preserve stronger links between the voltage variables through convex
  envelopes of the polar representation.  This paper is a systematic
  study of the QC relaxation for AC Optimal Power Flow with realistic
  side constraints. The main theoretical result shows that the QC
  relaxation is stronger than the SOC relaxation and neither dominates
  nor is dominated by the SDP relaxation. In addition, comprehensive
  computational results show that the QC relaxation may produce
  significant improvements in accuracy over the SOC relaxation at a
  reasonable computational cost, especially for networks with tight
  bounds on phase angle differences. The QC and SOC relaxations are
  also shown to be significantly faster and reliable compared to the
  SDP relaxation given the current state of the respective solvers.
\end{abstract}


\IEEEpeerreviewmaketitle

\vspace{-0.2cm}
\section*{Nomenclature}
\addcontentsline{toc}{section}{Nomenclature}
\begin{IEEEdescription}[\IEEEusemathlabelsep\IEEEsetlabelwidth{$Y^s = g^s + \bm i$}]
  \item [{$N$}]  - The set of nodes in the network 
  \item [{$E$}]  - The set of {\em from} edges in the network 
  \item [{$E^R$}]  - The set of {\em to} edges in the network 
  %
  \item [{$\bm i$}] - imaginary number constant
  \item [{$I$}] - AC current
  \item [{$S = p+ \bm iq$}] - AC power
  \item [{$V = v \angle \theta$}]  - AC voltage
  \item [{$Z = r+ \bm ix$}] - Line impedance
  \item [{$Y = g + \bm ib$}]  - Line admittance
  \item [{$T = t \angle \theta^t$}]  - Transformer properties
  \item [{$Y^s = g^s + \bm ib^s$}]  - Bus shunt admittance
  \item [{$W $}]  - Product of two AC voltages
  \item [{$l$}]  - Current magnitude squared, $|I|^2$
  %
  \item [{$b^c$}] - Line charging
  \item [{$s^u$}] - Line apparent power thermal limit
  \item [{$\theta^\Delta$}] - Phase angle difference limit
  \item [{$S^d = p^d+ \bm iq^d$}] - AC power demand
  \item [{$S^g = p^g+ \bm iq^g$}] - AC power generation
  \item [{$c_0,c_1,c_2$}] - Generation cost coefficients 
 %
   \item [{$\Re(\cdot)$}] - Real part of a complex number
   \item [{$\Im(\cdot)$}] - Imaginary part of a complex number
   \item [{$(\cdot)^*$}] - Conjugate of a complex number
   \item [{$|\cdot|$}] - Magnitude of a complex number, $l^2$-norm
  %
  %
  \item [{$x^l, x^u$}] - Lower and upper bounds of $x$, respectively
  \item [{$\widecheck{x}$}] - Convex envelope of $x$
  \item [{$\bm x$}] - A constant value
\end{IEEEdescription}

\section{Introduction}
\label{sec:intro}

\IEEEPARstart{C}{onvex} relaxations of the power flow equations have attracted
significant interest in recent years. They include the Semi-Definite
Programming (SDP) \cite{Bai2008383}, Second-Order Cone (SOC)
\cite{Jabr06}, Convex-DistFlow (CDF) \cite{6102366}, and the recent
Quadratic Convex (QC) \cite{QCarchive} and Moment-Based \cite{7038397,6980142} relaxations. Much of the
excitement underlying this line of research comes from the fact that
the SDP relaxation has shown to be tight on a variety of case studies
\cite{5971792}, opening a new avenue for accurate, reliable, and
efficient solutions to a variety of power system applications. Indeed,
industrial-strength optimization tools (e.g., Gurobi,
cplex, Mosek) are now available to solve
various classes of convex optimization problems. 

The relationships between the SDP, SOC, and CDF relaxations is now
largely well-understood: See \cite{6756976,6815671} for a
comprehensive overview. In particular, the SOC and CDF relaxations are
known to be equivalent and the SDP relaxation is at least as strong
than both of these. However, little is known about the QC relaxation
which is a significant departure from these more traditional
relaxations. Indeed, one of the key features of the QC relaxation is
to compute convex envelopes of the polar representation of the power
flow equations in the hope of preserving stronger links between the
voltage variables. This contrasts with the SDP and SOC relaxations
which are derived from a lift-and-project approach on the complex
representation.

This paper fills this gap and provides a theoretical study of the QC
relaxation as well as a comprehensive computational evaluation to compare
the strengths and weaknesses of these relaxations. Our main
contributions can be summarized as follows:

\begin{enumerate}
\item The QC relaxation is stronger than the SOC relaxation.

\item The QC relaxation neither dominates nor is dominated by the SDP
  relaxation.

\item Computational results on optimal power flow show that the QC
  relaxation may bring significant benefits in accuracy over the SOC
  relaxation, especially for tight bounds on phase angle differences,
  for a reasonable loss in efficiency.

\item The computational results also show that, with existing
  solvers, the SOC and QC relaxations are significantly faster and more
  reliable than the SDP relaxation.
\end{enumerate}

\noindent
The theoretical results are derived using the equivalence of two
classes of second-order cone constraints (in conjunction with the
power equations), which provides an alternative formulation for the QC
model which is interesting in its own right. Moreover, to the best of
our knowledge, the computational results also represent the most
comprehensive comparison of these convex relaxations. They are
obtained for optimal power flow problems with realistic
side-constraints, featuring bus shunts, line charging, and
transformers.

The rest of the paper is organized as follows.  Section
\ref{sec:ac:pf} reviews the formulation of the AC-OPF problem from
first principles and presents two equivalent formulations of this
non-convex optimization problem.  Section \ref{sec:relaxations}
derives the SDP, QC, and SOC relaxations. Section \ref{sec:example}
illustrates their behavior on a well-known 3-bus example. Section
\ref{sec:qc_alt} presents an alternative formulation of the QC
relaxation which is a convenient tool for subsequent proofs. Section
\ref{sec:relations} presents the theoretical results linking the QC to
the other relaxations.  Section \ref{sec:expar} reports the
computational results for the three relaxations on 93 AC-OPF test
cases, and Section \ref{sec:conclusion} concludes the paper.

\section{AC Optimal Power Flow}
\label{sec:ac:pf}

This section reviews the specification of AC Optimal Power Flow (AC-OPF)
and introduces the notations used in the paper. In the equations,
constants are always in bold face. The AC power flow equations are
based on complex quantities for current $I$, voltage $V$, admittance
$Y$, and power $S$, which are linked by the physical properties of
Kirchhoff's Current Law (KCL), i.e.,
\begin{align}
& I^g_i - {\bm I^d_i} = \sum_{\substack{(i,j)\in E \cup E^R}} I_{ij} 
\end{align}
Ohm's Law, i.e.,
\begin{align}
& I_{ij} = \bm Y_{ij} (V_i - V_j)  \label{current_flow}
\end{align}
and the definition of AC power, i.e.,
\begin{align}
& S_{ij} = V_{i}I_{ij}^* \label{complex_power}
\end{align}
Combining these three properties yields the AC Power Flow equations, i.e.,
\begin{subequations}
\begin{align}
& S^g_i - {\bm S^d_i} = \sum_{\substack{(i,j)\in E \cup E^R}} S_{ij} \;\; \forall i\in N \\ 
& S_{ij} = \bm Y^*_{ij} V_i V^*_i - \bm Y^*_{ij} V_i V^*_j \;\; (i,j)\in E \cup E^R
\end{align}
\end{subequations}

\noindent
These non-convex nonlinear equations define how power flows in the
network and are a core building block in many power system
applications. However, practical applications typically include various
operational side constraints on the power flow. We now review some 
of the most significant ones.

\paragraph*{Generator Capabilities}

AC generators have limitations on the amount of active and reactive
power they can produce $S^g$, which is characterized by a generation
capability curve \cite{9780070359581}.  Such curves typically define
nonlinear convex regions which are typically approximated by boxes in
AC transmission system test cases, i.e.,
\begin{subequations}
\begin{align}
& \bm {S^{gl}}_i \leq S^g_i \leq \bm {S^{gu}}_i \;\; \forall i \in N 
\end{align}
\end{subequations}

\paragraph*{Line Thermal Limit}

AC power lines have thermal limits \cite{9780070359581} to prevent
lines from sagging and automatic protection devices from activating.
These limits are typically given in Volt Amp units and constrain 
the apparent power flows on the lines, i.e.,
\begin{align}
& |S_{ij}| \leq \bm {s^u}_{ij} \;\; \forall (i,j) \in E \cup E^R 
\end{align}

\paragraph*{Bus Voltage Limits}

Voltages in AC power systems should not vary too far (typically $\pm
10\%$) from some nominal base value \cite{9780070359581}.  This is
accomplished by putting bounds on the voltage magnitudes, i.e.,
\begin{align}
& \bm {v^l}_i \leq |V_i| \leq \bm {v^u}_i \;\; \forall i \in N
\end{align}
A variety of power flow formulations only have variables for the
square of the voltage magnitude, i.e., $|V_i|^2$.  In such cases, the
voltage bound constrains can be incorporated via the following constraints:
\begin{align}
& ( \bm {v^l}_{i} )^2 \leq |V_i|^2 \leq ( \bm {v^u}_{i} )^2 \;\; \forall i \in N
\end{align}

\paragraph*{Phase Angle Differences}

Small phase angle differences are also a design imperative in AC power
systems \cite{9780070359581} and it has been suggested that phase
angle differences are typically less than $10$ degrees in practice
\cite{Purchala:2005gt}. These constraints have not typically been
incorporated in AC transmission test cases \cite{matpower}. However,
recent work \cite{LPAC_ijoc,QCarchive} have observed that
incorporating Phase Angle Difference (PAD) constraints, i.e.,
\begin{align}
&  -\bm {\theta^\Delta}_{ij} \leq \angle \! \left( V_i V^*_j \right) \leq \bm {\theta^\Delta}_{ij} \;\; \forall (i,j) \in E \label{eq:pad_1}
\end{align}
is useful in the convexification of the AC power flow equations. For
simplicity, this paper assumes that the phase angle difference bounds
are symmetrical and within the range $(- \pi/2, \pi/2 )$, i.e.,
\begin{align}
& 0 \leq \bm {\theta^{\Delta}}_{ij} \leq \frac{\pi}{2} \;\; (i,j) \in E
\end{align}
but the results presented here can be extended to more general
cases. Observe also that the PAD constraints \eqref{eq:pad_1} can be
implemented as a linear relation of the real and imaginary components
of $V_iV^*_j$ \cite{6810520}, i.e. $\forall (i,j) \in E$,
\begin{align}
& \tan(-\bm {\theta^\Delta}_{ij}) \Re\left(V_iV^*_j\right) \! \leq \!  \Im\left(V_iV^*_j\right) \! \leq \! \tan(\bm {\theta^\Delta}_{ij}) \Re\left(V_iV^*_j\right) \label{eq:w_pad}
\end{align}
The usefulness of this formulation will be apparent later in the
paper.

\paragraph*{Other Constraints}
Other line flow constraints have been proposed, such as, active power
limits and voltage difference limits \cite{5971792,6810520}.  However,
we do not consider them here since, to the best of our knowledge, test
cases incorporating these constraints are not readily available.

\paragraph*{Objective Functions}

The last component in formulating OPF problems is an objective
function. The two classic objective functions are line loss
minimization, i.e.,
\begin{align}
& \mbox{minimize: } \sum_{i \in N} \Re(S^g_i)  \label{eq:loss_min}
\end{align}
and generator fuel cost minimization, i.e.,
\begin{align}
& \mbox{minimize: } \sum_{i \in N} \bm c_{2i} (\Re(S^g_i))^2 + \bm c_{1i}\Re(S^g_i) + \bm c_{0i} \label{eq:fule_min}
\end{align}
Observe that objective \eqref{eq:loss_min} is a special case of
objective \eqref{eq:fule_min} where $\bm c_{2i}\!=\!0, \bm
c_{1i}\!=\!1, \bm c_{0i}\!=\!0 \;\; (i \! \in \! N)$ \cite{6153415}.
Hence, the rest of this paper focuses on objective
\eqref{eq:fule_min}.

\paragraph*{AC-OPF}

Combining the AC power flow equations, the side constraints, and the
objective function, yields the well-known AC-OPF formulation presented
in Model \ref{model:ac_opf}. Observe that, in Model
\ref{model:ac_opf}, the non-convexities arises solely from the product
of the voltages (i.e., $V_i V_j^*$) and they can be isolated by
introducing new $W$ variables to represent the products of $V$s
\cite{780924, Jabr06, 4548149, 6345272}, i.e,
\begin{align}
& V_i V_j^* = W_{ij} \;\; (i,j \in N).
\end{align}
Model \ref{model:ac_opf_w} presents an equivalent version of the
AC-OPF, where the $W$ factorization has been incorporated and the only
source of non-convexity is in constraint \eqref{w_2}.  
Note that this section has introduced the simplest form of the AC-OPF problem and that real-world applications feature a variety of extensions as discussed at length in \cite{Capitanescu20111731,real_opf}.  In practice, this non-convex nonlinear optimization problem is typically solved with numerical methods (e.g. IPM, SLP) \cite{744492,744495}, which provide locally optimal solutions if they converge to a feasible point.

\begin{model}[t]
\caption{AC-OPF}
\label{model:ac_opf}
\begin{subequations}
\vspace{-0.2cm}
\begin{align}
\mbox{\bf variables: } & S^g_i (\forall i\in N), \; V_i (\forall i\in N)  \nonumber \\
%
\mbox{\bf minimize: } & \sum_{i \in N} \bm c_{2i} (\Re(S^g_i))^2 + \bm c_{1i}\Re(S^g_i) + \bm c_{0i} \label{ac_obj} \\
\mbox{\bf subject to:} \nonumber
\end{align}
\vspace{-0.7cm}
\begin{align}
\phantom{1234} & \bm {v^l}_i \leq |V_i| \leq \bm {v^u}_i \;\; \forall i \in N \label{ac_1} \\
& \bm {S^{gl}}_i \leq S^g_i \leq \bm {S^{gu}}_i \;\; \forall i \in N \label{ac_2}  \\
& |S_{ij}| \leq \bm {s^u}_{ij} \;\; \forall (i,j) \in E \cup E^R \label{ac_5}  \\
& S^g_i - {\bm S^d_i} = \sum_{\substack{(i,j)\in E \cup E^R}} S_{ij} \;\; \forall i\in N \label{ac_3}  \\ 
& S_{ij} = \bm Y^*_{ij} V_i V^*_i - \bm Y^*_{ij} V_i V^*_j \;\; (i,j)\in E \cup E^R \label{ac_4}  \\
& -\bm {\theta^\Delta}_{ij} \leq \angle (V_i V^*_j) \leq \bm {\theta^\Delta}_{ij} \;\; \forall (i,j) \in E  \label{ac_6} 
\end{align}
\end{subequations}
\end{model}

\begin{model}[t]
\caption{ AC-OPF-W}
\label{model:ac_opf_w}
\begin{subequations}
\vspace{-0.2cm}
\begin{align}
&\mbox{\bf variables: } S^g_i (\forall i\in N), \; V_i (\forall i\in N), \; W_{ij} (\forall i,j \in N) \phantom{123}  \nonumber \\
%
&\mbox{\bf minimize: } \sum_{i \in N} \bm c_{2i} (\Re(S^g_i))^2 + \bm c_{1i}\Re(S^g_i) + \bm c_{0i} \label{w_obj} \\
&\mbox{\bf subject to: } \nonumber 
\end{align}
\vspace{-0.7cm}
\begin{align}
%
\phantom{123} &  W_{ij} = V_iV_j^* \;\; \forall i \in N,  \forall j \in N \label{w_2} \\
& (\bm {v^l}_i)^2 \leq W_{ii} \leq (\bm {v^u}_i)^2 \;\; \forall i \in N \label{w_3} \\
& \bm {S^{gl}}_i \leq S^g_i \leq \bm {S^{gu}}_i \;\; \forall i \in N \label{w_4} \\
& S^g_i - {\bm S^d_i} = \sum_{\substack{(i,j)\in E \cup E^R}} S_{ij} \;\; \forall i\in N \label{w_5} \\ 
& S_{ij} = \bm Y^*_{ij} W_{ii} - \bm Y^*_{ij} W_{ij} \;\; (i,j)\in E \label{w_6} \\
& S_{ji} = \bm Y^*_{ij} W_{jj} - \bm Y^*_{ij} W_{ij}^* \;\; (i,j)\in E \label{w_7} \\
& |S_{ij}| \leq (\bm {s^u}_{ij}) \;\; \forall (i,j) \in E \cup E^R \label{w_8} \\
& \tan(-\bm {\theta^\Delta}_{ij}) \Re(W_{ij}) \leq \Im(W_{ij}) \leq \tan(\bm {\theta^\Delta}_{ij}) \Re(W_{ij}) \label{w_9} \\
& \hspace{5.5cm} \forall (i,j) \in E \nonumber
\end{align}
\end{subequations}
\end{model}

\section{Convex Relaxations of Optimal Power Flow}
\label{sec:relaxations}

Since the AC-OPF problem is NP-Hard \cite{verma2009power,ACSTAR2015} 
and numerical methods provide limited guarantees for determining feasibility and global optimally,
significant attention has been devoted to finding convex relaxations
of Model \ref{model:ac_opf}.  Such relaxations are appealing because
they are computationally efficient and may be used to:
\begin{enumerate}
\item bound the quality of AC-OPF solutions produced by locally optimal methods;
\item prove that a particular AC-OPF problem has no solution; 
\item produce a solution that is feasible in the original non-convex
  problem \cite{5971792}, thus solving the AC-OPF and guaranteeing
  that the solution is globally optimal.
\end{enumerate}
The ability to provide bounds is particularly important for the
numerous mixed-integer nonlinear optimization problems that arise in power
system applications.
For these reasons, a variety of convex relaxations of the AC-OPF have
been developed including, the SDP \cite{Bai2008383}, QC
\cite{QCarchive}, SOC \cite{Jabr06}, and Convex-DistFlow
\cite{6102366}, which are reviewed in detail in this
section. Moreover, since the SOC and Convex-DistFlow relaxations have
been shown to be equivalent \cite{6483453}, this paper focuses on the
SDP, SOC, and QC relaxations only and shows how they are derived from
Model \ref{model:ac_opf_w}.  The key insight is that each relaxation
presents a different approach to convexifing constraints \eqref{w_2},
which are the only source of non-convexity in Model
\ref{model:ac_opf_w}.

\paragraph*{The Semi-Definite Programming  (SDP) Relaxation} 

exploits the fact that the $W$ variables are defined by $V(V^*)^T$,
which ensures that $W$ is positive semi-definite (denoted by $W
\succeq 0$) and has rank 1 \cite{Bai2008383, 5971792, 6345272}.  These
conditions are sufficient to enforce constraints \eqref{w_2}
\cite{doi:10.1137/1038003}, i.e.,
\begin{equation}
W_{ij} = V_iV_j^* \; (i,j \in N) \;\; \Leftrightarrow \;\; W \succeq 0 \; \wedge \; \mbox{rank}(W) = 1 \nonumber
\end{equation}
The SDP relaxation \cite{sdpIntro,doi:10.1137/1038003} then drops the
rank constraint to obtain Model \ref{model:ac_opf_w_sdp}.

\begin{model}[t]
\caption{The SDP Relaxation AC-OPF-W-SDP.}
\label{model:ac_opf_w_sdp}
\begin{subequations}
\vspace{-0.2cm}
\begin{align}
\mbox{\bf variables: } & S^g_i (\forall i\in N), \; W_{ij} (\forall i,j \in N)  \nonumber \\
%
\mbox{\bf minimize: } & \eqref{w_obj} \nonumber \\
\mbox{\bf subject to: } & \mbox{\eqref{w_3}--\eqref{w_9}} \nonumber \\
& W \succeq 0 \label{w_sdp}
\end{align}
\end{subequations}
\end{model}

\begin{model}[t]
\caption{The SOC Relaxation AC-OPF-W-SOC.}
\label{model:ac_opf_w_soc}
\begin{subequations}
\vspace{-0.2cm}
\begin{align}
\mbox{\bf variables: } & S^g_i (\forall i\in N), \; W_{ij} (\forall (i,j)\in E), \; W_{ii} (\forall i \in N) : real   \nonumber \\
%
\mbox{\bf minimize: } & \eqref{w_obj} \nonumber \\
\mbox{\bf subject to: } & \mbox{\eqref{w_3}--\eqref{w_9}} \nonumber \\
& |W_{ij}|^2 \leq W_{ii}W_{jj} \;\; \forall (i,j)\in E \label{soc_1}
\end{align}
\end{subequations}
\end{model}

\paragraph*{The Second Order Cone (SOC) Relaxation} convexifies each
constraint of \eqref{w_2} separately, instead of considering them
globally as in the SDP relaxation. The SOC relaxation takes the
absolute square of each constraint, refactors it, and then relaxes the
equality into an inequality, i.e.,
\begin{subequations}
\begin{align}
& W_{ij} = V_iV^*_j \\
& W_{ij}W^*_{ij} = V_iV^*_jV^*_iV_j \\
& |W_{ij}|^2 = W_{ii}W_{jj} \\
& |W_{ij}|^2 \leq W_{ii}W_{jj} \label{w_soc}
\end{align}
\end{subequations}
%
Equation \eqref{w_soc} is a rotated second-order cone constraint which
is widely supported by industrial optimization tools. It can, in fact,
be rewritten in the standard form of a second-order cone constraint as,
\begin{equation}
 \left | \begin{pmatrix}
  2W_{ij}\\
  W_{ii} - W_{jj}\\
 \end{pmatrix} \right | \leq W_{ii} + W_{jj}  
\end{equation}
The complete SOC formulation is presented in Model
\ref{model:ac_opf_w_soc}.  Note that this relaxation requires fewer
$W$ variables than Model \ref{model:ac_opf_w_sdp}. Due to the
sparsity of AC power networks, this size reduction can lead to
significant memory and computational savings.

\paragraph*{The Quadratic Convex (QC) Relaxation}
was introduced to preserve stronger links between the voltage
variables \cite{QCarchive}.  It represents the voltages in polar
form (i.e., $V = v \angle \theta$) and links these real variables
to the $W$ variables, along the lines of \cite{780924,4548149,6661462,RomeroRamos2010562}, using the following equations:
\begin{subequations}
\begin{align}
& W_{ii} = v_{i}^2 \;\; i \in N \label{eq:w_link_1} \\
& \Re(W_{ij}) = v_{i}v_{j}\cos(\theta_i - \theta_j) \;\; \forall(i,j) \in E \label{eq:w_link_2} \\
& \Im(W_{ij}) = v_{i}v_{j}\sin(\theta_i - \theta_j) \;\; \forall(i,j) \in E \label{eq:w_link_3}
\end{align}
\end{subequations}
The QC relaxation then relaxes these equations by taking tight convex
envelopes of their nonlinear terms, exploiting the operational limits
for $v_i, v_j, \theta_i - \theta_j$. The convex envelopes for the
square and product of variables are well-known \cite{MacC76}, i.e.,
\begin{equation}
\tag{T-CONV}
\langle x^2 \rangle^T \equiv
\begin{cases*}
\widecheck{x}  \geq  x^2\\
\widecheck{x}  \leq  ( \bm {x^u} + \bm {x^l})x - \bm {x^u} \bm {x^l}
\end{cases*}
\end{equation}
\begin{equation*}
\tag{M-CONV}
\langle xy \rangle^M \equiv
\begin{cases*}
\widecheck{xy}  \geq  \bm {x^l}y + \bm {y^l}x - \bm {x^l}\bm {y^l}\\
\widecheck{xy}  \geq  \bm {x^u}y + \bm {y^u}x - \bm {x^u}\bm {y^u}\\
\widecheck{xy}  \leq  \bm {x^l}y + \bm {y^u}x - \bm {x^l}\bm {y^u}\\
\widecheck{xy}  \leq  \bm {x^u}y + \bm {y^l}x - \bm {x^u}\bm {y^l}
\end{cases*}
\end{equation*}
Under our assumptions that the phase angle bound satisfies $0 \leq \bm
{\theta^\Delta} \leq \frac{\pi}{2}$ and is symmetric, convex envelopes
for sine (S-CONV) and cosine (C-CONV) \cite{QCarchive} are given by,
\begin{equation*}
\langle \sin(x) \rangle^S \equiv
\begin{cases*}
\widecheck{sx} \leq \cos\left(\frac{\bm {x^u}}{2}\right)\left(x -\frac{\bm {x^u}}{2}\right) + \sin\left(\frac{\bm {x^u}}{2}\right)\\
\widecheck{sx} \geq \cos\left(\frac{\bm {x^u}}{2}\right)\left(x +\frac{\bm {x^u}}{2}\right) - \sin\left(\frac{\bm {x^u}}{2}\right)
\end{cases*}
\end{equation*}
\begin{equation*}
\langle \cos(x) \rangle^C \equiv
\begin{cases*}
\widecheck{cx}  \leq  1 - \frac{1-\cos({\bm {x^u}})}{({\bm {x^u}})^2} x^2\\
\widecheck{cx}  \geq \cos(\bm {x^u})
\end{cases*}
\end{equation*}
\noindent
In the following, we abuse notation and also use $\langle f(\cdot)
\rangle^{C}$ to denote the variable on the left-hand side of the
convex envelope $C$ for function $f(\cdot)$. When such an expression is used
inside an equation, the constraints $\langle f(\cdot) \rangle^{C}$ are
also added to the model.

\begin{model}[t]
\caption{The QC Relaxation AC-OPF-C-QC.}
\label{model:ac_opf_c_qc}
\begin{subequations}
\vspace{-0.2cm}
\begin{align}
\mbox{\bf variables: } & S^g_i (\forall i\in N), \; W_{ij} (\forall (i,j)\in E), \; W_{ii} (\forall i \in N) : real   \nonumber \\
& v_i \angle \theta_i (\forall i \in N), \;  l_{ij} (\forall (i,j)\in E) \nonumber \\
%
\mbox{\bf minimize: } & \eqref{w_obj} \nonumber \\
\mbox{\bf subject to: } & \mbox{\eqref{w_3}--\eqref{w_9}} \nonumber 
\end{align}
\vspace{-0.7cm}
\begin{align}
\phantom{1} & W_{ii} = \langle v_i^2 \rangle^T  \;\; i \in N \label{qc_1} \\
&\Re(W_{ij}) = \langle \langle v_i v_j \rangle^M \langle \cos(\theta_i - \theta_j) \rangle^C \rangle^M \;\; \forall(i,j) \in E \label{qc_2} \\
&\Im(W_{ij}) = \langle \langle v_i v_j \rangle^M \langle \sin(\theta_i - \theta_j) \rangle^S \rangle^M  \;\; \forall(i,j) \in E \label{qc_3} \\
& S_{ij} + S_{ji} = \bm Z_{ij} l_{ij} \;\; \forall(i,j) \in E \label{qc_4} \\
& |S_{ij}|^2 \leq W_{ii}l_{ij}  \;\; \forall(i,j) \in E \label{qc_5}
\end{align}
\end{subequations}
\end{model}

Convex envelopes for equations
\eqref{eq:w_link_1}--\eqref{eq:w_link_3} can be obtained by composing
the convex envelopes of the functions for square, sine, cosine, and
the product of two variables, i.e.,
\begin{subequations}
\begin{align}
& W_{ii} = \langle v_i^2 \rangle^T  \;\; i \in N \\
&\Re(W_{ij}) = \langle \langle v_i v_j \rangle^M \langle \cos(\theta_i - \theta_j) \rangle^C \rangle^M \;\; \forall(i,j) \in E \\
&\Im(W_{ij}) = \langle \langle v_i v_j \rangle^M \langle \sin(\theta_i - \theta_j) \rangle^S \rangle^M  \;\; \forall(i,j) \in E 
\end{align}
\end{subequations}
The QC relaxation also proposes to strengthen these convex envelopes
with a second-order cone constraint based on the absolute square of line 
power flow \eqref{complex_power}, first proposed in \cite{6102366}. This requires a new variable $l_{ij}$ for each
line $(i,j) \in E$ that captures the current magnitude squared on that line.
The following constraints are added to link the $l_{ij}$ variables to
the existing model variables.
\begin{subequations}
\begin{align}
& S_{ij} + S_{ji} = \bm Z_{ij} l_{ij} \;\; \forall(i,j) \in E \\
& |S_{ij}|^2 \leq W_{ii}l_{ij}  \;\; \forall(i,j) \in E 
\end{align}
\end{subequations}

\noindent
The complete QC relaxation is presented in Model
\ref{model:ac_opf_c_qc}.  This model is annotated as C-QC, as the
second-order cone constraints use current variables.  The motivation
for this distinction will become clear in Section \ref{sec:qc_alt}.

\section{An Illustrative Example}
\label{sec:example}

\begin{figure}[t]
\centering
\vspace{-0.2cm}
    \includegraphics[width=2.5cm]{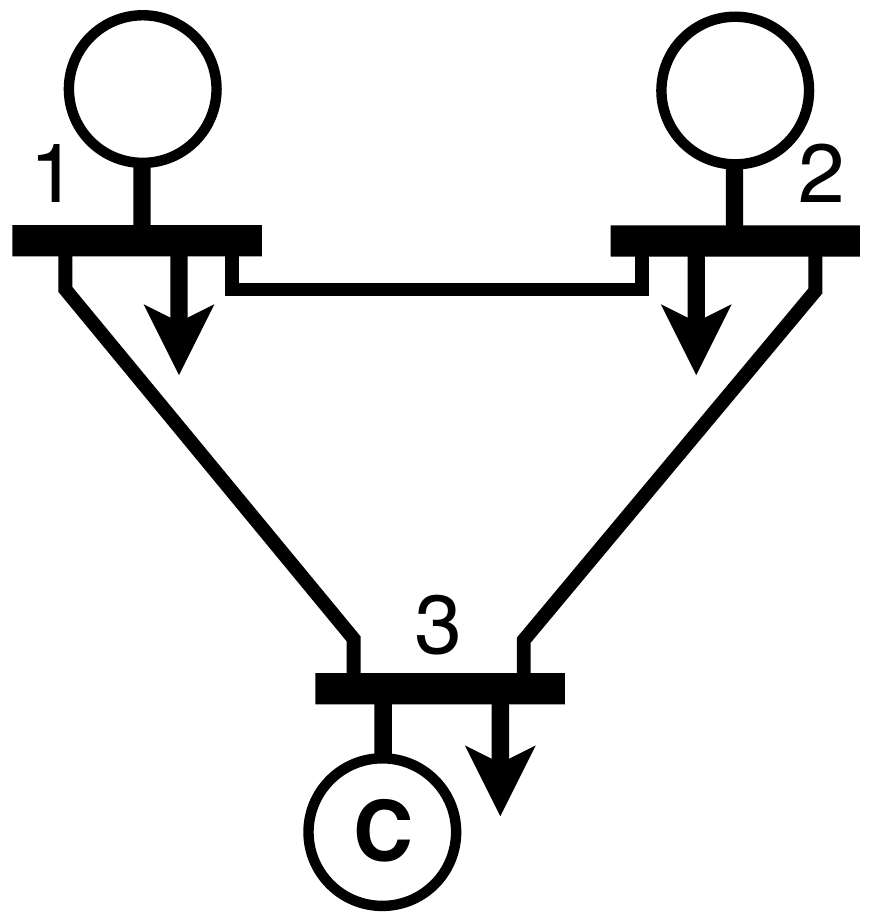} 
\vspace{-0.2cm}
\caption{A Network Diagram for the 3-Bus Example.}
\label{fig:3_bus_example}
\end{figure}

\begin{table}[t]
\caption{Three-Bus System Network Data (100 MVA Base).}
\centering
\begin{tabular}{|c||c|c|c|c|}
\hline
\multicolumn{5}{|c|}{Bus Parameters} \\
\hline
Bus & $\boldsymbol {p^d}$ & $\boldsymbol {q^d}$ & $\boldsymbol {v^l}$ & $\boldsymbol {v^u}$ \\
\hline
1 & 110 & 40 & 0.9 & 1.1 \\ 
\hline 
2 & 110 & 40 & 0.9 & 1.1 \\ 
\hline 
3 & 95 & 50 & 0.9 & 1.1 \\ 
\hline 
\end{tabular}
\vspace{0.1cm}

\begin{tabular}{|c||c|c|c|c|c|}
\hline
\multicolumn{6}{|c|}{Line Parameters} \\
\hline
From--To Bus & $\boldsymbol {r}$ & $\boldsymbol {x}$ & $\boldsymbol {b^c}$ & $\boldsymbol {s^u}$ & $\bm {\theta^\Delta}$ \\
\hline
1--2 & 0.042 & 0.90 & 0.30 & $\infty$ & $30^\circ$ \\ 
\hline 
2--3 & 0.025 & 0.75 & 0.70 & 50 & $30^\circ$ \\ 
\hline 
1--3 & 0.065 & 0.62 & 0.45 & $\infty$ & $30^\circ$ \\ 
\hline 
\end{tabular}
\vspace{0.1cm}

\begin{tabular}{|c||c|c||c|c|c|}
\hline
\multicolumn{6}{|c|}{Generator Parameters} \\
\hline
Generator & $\boldsymbol {p^{gl}}, \boldsymbol {p^{gu}}$ & $\boldsymbol {q^{gl}}, \boldsymbol {q^{gu}}$ & $\boldsymbol {c_2}$ & $\boldsymbol {c_1}$ & $\boldsymbol {c_0}$ \\
\hline
1 & $0,\infty$ & $-\infty,\infty$ & 0.110 & 5.0 & 0 \\ 
\hline 
2 & $0,\infty$ & $-\infty,\infty$ & 0.085 & 1.2 & 0 \\ 
\hline 
3 & $0,0$ & $-\infty,\infty$ & 0 & 0 & 0 \\ 
\hline 
\end{tabular}
\label{tbl:3_bus_network_data}
\end{table}

This section illustrates the three main power flow relaxations on the
3-bus network from \cite{6120344}, which has proven to be an excellent
test case for power flow relaxations. This system is depicted in
Figure \ref{fig:3_bus_example} and the associated network parameters
are given in Table \ref{tbl:3_bus_network_data}.  This network is
designed to have very few binding constraints.  Hence, the generator
and line limits are set to large non-binding values, except for the
thermal limit constraint on the line between buses 2 and 3, which is
set to 50 MVA. In addition to its base configuration, we also consider
this network with reduced phase angle difference bounds of
$18^\circ$. IPOPT \cite{Ipopt} is used as a heuristic \cite{6581918}
to find a feasible solution to the AC-OPF and we measure the {\em
  optimally gap} between the heuristic and a relaxation using the
formula
\begin{align}
\frac{\text{Heuristic} - \text{Relaxation}}{\text{Heuristic}}. \nonumber
\end{align}

\begin{table}[t]
\caption{AC-OPF Bounds using Relaxations on the 3-Bus Case.}
\centering
\begin{tabular}{|r||r||r|r|r|r|r|r|r|r|r||r|r|r|r|c|c|}
\hline
                 & \$/h & \multicolumn{3}{c|}{Optimality Gap (\%)}  \\
Test Case & AC  & SDP & QC & SOC \\
\hline  
\hline  
Base & {\bf 5812} & 0.39 & 1.24 & 1.32  \\
\hline 
 $ {\theta^\Delta} \! = \!18^\circ$ & {\bf 5992} & 2.06 & 1.24 & 4.28 \\
\hline 
\end{tabular}
\label{tbl:3_bus_results}
\end{table}


\noindent
Table \ref{tbl:3_bus_results} summarizes the results.\footnote{On this
  small example a nonlinear global optimization solver was used to
  prove that the heuristic solutions are in fact globally optimal.
  Such a validation is not possible on larger test cases.} In the base
configuration, the SDP relaxation has the smallest optimality gap. In
the $ {\theta^\Delta} \! = \!18^\circ$ case, the QC relaxation has the
smallest optimality gap, while reducing the bound on phase angle
differences increases the optimality gap for both the SDP and SOC
relaxations. This small network highlights two important results. First, the SDP
relaxation does not dominate the QC relaxation and vice-versa. Second,
the SDP and QC relaxations dominate the SOC relaxation. The next two
sections prove that this last result holds for all networks.

\section{An Alternate Form of the QC Relaxation}
\label{sec:qc_alt}

Section \ref{sec:relaxations} introduced two types of second-order
cone constraints.  Model \ref{model:ac_opf_w_soc} uses a SOC
constraint based on the absolute square of the voltage product
\cite{Jabr06}, i.e.,
\begin{align}
& |W_{ij}|^2 \leq W_{ii}W_{jj}  
\end{align}
while Model \ref{model:ac_opf_c_qc} uses a SOC constraint based on the
absolute square of the power flow \cite{6102366}, i.e.,
\begin{align}
& |S_{ij}|^2 \leq W_{ii}l_{ij}.
\end{align}
We now show that, in conjunction with the power flow equations
\eqref{w_6}--\eqref{w_7}, these two SOC formulations are equivalent. More precisely,
we show that 
\begin{equation}
\label{soc_w} \tag{W-SOC}
\begin{array}{ll}
& S_{ij} = \bm Y^*_{ij} W_{ii} - \bm Y^*_{ij} W_{ij} \;\; (i,j)\in E \\
& S_{ji} = \bm Y^*_{ij} W_{jj} - \bm Y^*_{ij} W_{ij}^* \;\; (i,j)\in E \\
& |W_{ij}|^2 \leq W_{ii} W_{jj} \;\; (i,j)\in E  \\
\end{array}
\end{equation}
is equivalent to 
\begin{equation}
\label{soc_c} \tag{C-SOC}
\begin{array}{ll}
& S_{ij} = \bm Y^*_{ij} W_{ii} - \bm Y^*_{ij} W_{ij} \;\; (i,j)\in E \\
& S_{ji} = \bm Y^*_{ij} W_{jj} - \bm Y^*_{ij} W_{ij}^* \;\; (i,j)\in E \\
& S_{ij} + S_{ji} = \bm Z_{ij} l_{ij} \;\; (i,j) \in E \\
& |S_{ij}|^2 \leq W_{ii} l_{ij} \;\; (i,j)\in E.
\end{array}
\end{equation}

\begin{model}[t]
\caption{The Alternate QC Relaxation AC-OPF-W-QC}
\label{model:ac_opf_w_qc}
\begin{subequations}
\vspace{-0.2cm}
\begin{align}
\mbox{\bf variables: } & S^g_i (\forall i\in N), \; W_{ij} (\forall (i,j)\in E), \; W_{ii} (\forall i \in N) : real \nonumber \\
& v_i \angle \theta_i (\forall i \in N) \nonumber \\
\mbox{\bf minimize: } & \eqref{w_obj} \nonumber \\
\mbox{\bf subject to: } & \mbox{\eqref{w_3}--\eqref{w_9}, \eqref{qc_1}--\eqref{qc_3}, \eqref{soc_1}} \nonumber 
\end{align}
\end{subequations}
\end{model}

\noindent
This equivalence suggests an alternative formulation of the QC
relaxation which is given in Model \ref{model:ac_opf_w_qc} and
establishes a clear connection between Models \ref{model:ac_opf_w_soc}
and \ref{model:ac_opf_c_qc}. Throughout this paper, we use $W$ and $C$
to denote which of these equivalent formulations is used.

We now prove these results. The following lemma, whose proof is
straight-forward and can be found in the Appendix, establishes
some useful equalities.

\begin{lemma}
\label{lemma:properties}
The following four equalities hold:
\begin{enumerate}
\item $|S_{ij}|^2 = |\bm Y_{ij}|^2 \left( W_{ii}^2 - W_{ii} W_{ij} - W_{ii} W^*_{ij} + |W_{ij}|^2 \right)$.
\item $|W_{ij}|^2 = W_{ii}^2 - W^*_{ii} \bm Z^*_{ij} S_{ij} - W_{ii} \bm Z_{ij} S^*_{ij} + |\bm Z_{ij}|^2 |S_{ij}|^2$.
\item $l_{ij} = |\bm Y_{ij}|^2 (W_{ii} + W_{jj} - W_{ij} - W^*_{ij})$.
\item $W_{jj} = W_{ii} - \bm Z^*_{ij} S_{ij} - \bm Z_{ij} S^*_{ij} + |\bm Z_{ij}|^2 l_{ij}$.
\end{enumerate}
\end{lemma}

\noindent
We are now ready to prove the main result of this section. 

\begin{theorem}
\label{theorem:main}
\eqref{soc_c} is equivalent to \eqref{soc_w}.
\end{theorem}
\begin{proof}

The proof is similar in spirit to those presented in \cite{6483453,6897933}.  

\paragraph*{\ref{soc_w} $\Rightarrow$ \ref{soc_c}} Every solution to \eqref{soc_w} is a solution to \eqref{soc_c}. 
Given a solution to \eqref{soc_w}, by equality (3) in Lemma \ref{lemma:properties}, we assign $l_{ij}$ as follows:
\begin{align}
&  l_{ij} = |\bm Y_{ij}|^2 (W_{ii} - W_{ij} - W^*_{ij} + W_{jj}) \;\; (i,j) \in E \nonumber
\end{align}
This assignment satisfies the power loss constraint \eqref{qc_4} by
definition of the power. It remains to show that second-order cone
constraint in \eqref{soc_c} is satisfied. Using equalities (1) and (3)
in Lemma \ref{lemma:properties}, we obtain
\begin{subequations}
\begin{align}
& |S_{ij}|^2 = |\bm Y_{ij}|^2 \left( W_{ii}^2 - W_{ii} W_{ij} - W_{ii} W^*_{ij} + | W_{ij}|^2 \right) \nonumber \\
& |S_{ij}|^2 \leq |\bm Y_{ij}|^2 \left( W_{ii}^2 - W_{ii} W_{ij} -  W_{ii} W^*_{ij} + W_{ii} W_{jj} \right) \nonumber \\
& |S_{ij}|^2 \leq W_{ii} |\bm Y_{ij}|^2 \left(W_{ii} - W_{ij} - W^*_{ij} + W_{jj} \right) \nonumber \\
& |S_{ij}|^2 \leq W_{ii} l_{ij}. \nonumber
\end{align}
\end{subequations}

\paragraph*{\ref{soc_c} $\Rightarrow$ \ref{soc_w}} Every solution to \eqref{soc_c} is a solution to \eqref{soc_w}. 
We show that the values of $W_{ij}$ in \eqref{soc_c} satisfy the
second-order cone constraint in \eqref{soc_w}. Using equalities (2)
and (4) in Lemma \ref{lemma:properties} and the fact that $W_{ii} =
W^*_{ii}$ since $W_{ii}$ is a real number, we have
\begin{subequations}
\begin{align}
& |W_{ij}|^2 =  W_{ii}^2 -  W^*_{ii} \bm Z^*_{ij}  S_{ij} -  W_{ii} \bm Z_{ij}  S^*_{ij} + |\bm Z_{ij}|^2 | S_{ij}|^2 \nonumber \\
& | W_{ij}|^2 \leq  W_{ii}^2 -  W^*_{ii} \bm Z^*_{ij}  S_{ij} -  W_{ii} \bm Z_{ij}  S^*_{ij} + |\bm Z_{ij}|^2  W_{ii}  l_{ij} \nonumber \\
& | W_{ij}|^2 \leq  W_{ii}( W_{ii} - \bm Z^*_{ij}  S_{ij} - \bm Z_{ij}  S^*_{ij} + |\bm Z_{ij}|^2  l_{ij})  \nonumber \\
& | W_{ij}|^2 \leq  W_{ii} W_{jj} \nonumber
\end{align}
\end{subequations}
and the result follows.
\end{proof}

\begin{corollary}
\label{corollary:model}
Model \ref{model:ac_opf_c_qc} is equivalent to Model \ref{model:ac_opf_w_qc}.
\end{corollary}
\noindent
Computational results on these two formulations are presented in the Appendix. 
The main message is that the \ref{soc_c}
formulation is preferable to \ref{soc_w} in the current state of the
solving technology, especially on very large networks.

It is important to note that, for clarity, the proofs are presented on
the purest version of the AC power flow equations. Transmission system
test cases typically include additional parameters such as bus shunts,
line charging, and transformers.  Proofs that these results can be
extended to include the additional parameters in transmission system
test cases are presented in the Appendix. 

\section{Relations of the Power Flow Relaxations}
\label{sec:relations}

We are now in a position to state the relationships between the convex
relaxations.  Recall that model $M_1$ is a relaxation of model $M_2$,
denoted by $M_2 \subseteq M_1$, if the solution set of $M_2$ is
included in the solution set $M_1$. We use $M_1 \neq M_2$ to denote
the fact that neither $M_2 \subseteq M_1$ nor $M_1 \subseteq M_2$
holds. Since our relaxations have different sets of variables, we
define the solution set as the assignments to the $W_{ij}$ variables.

\begin{figure}[t]
\centering
\vspace{-0.2cm}
\includegraphics[width=4.0cm]{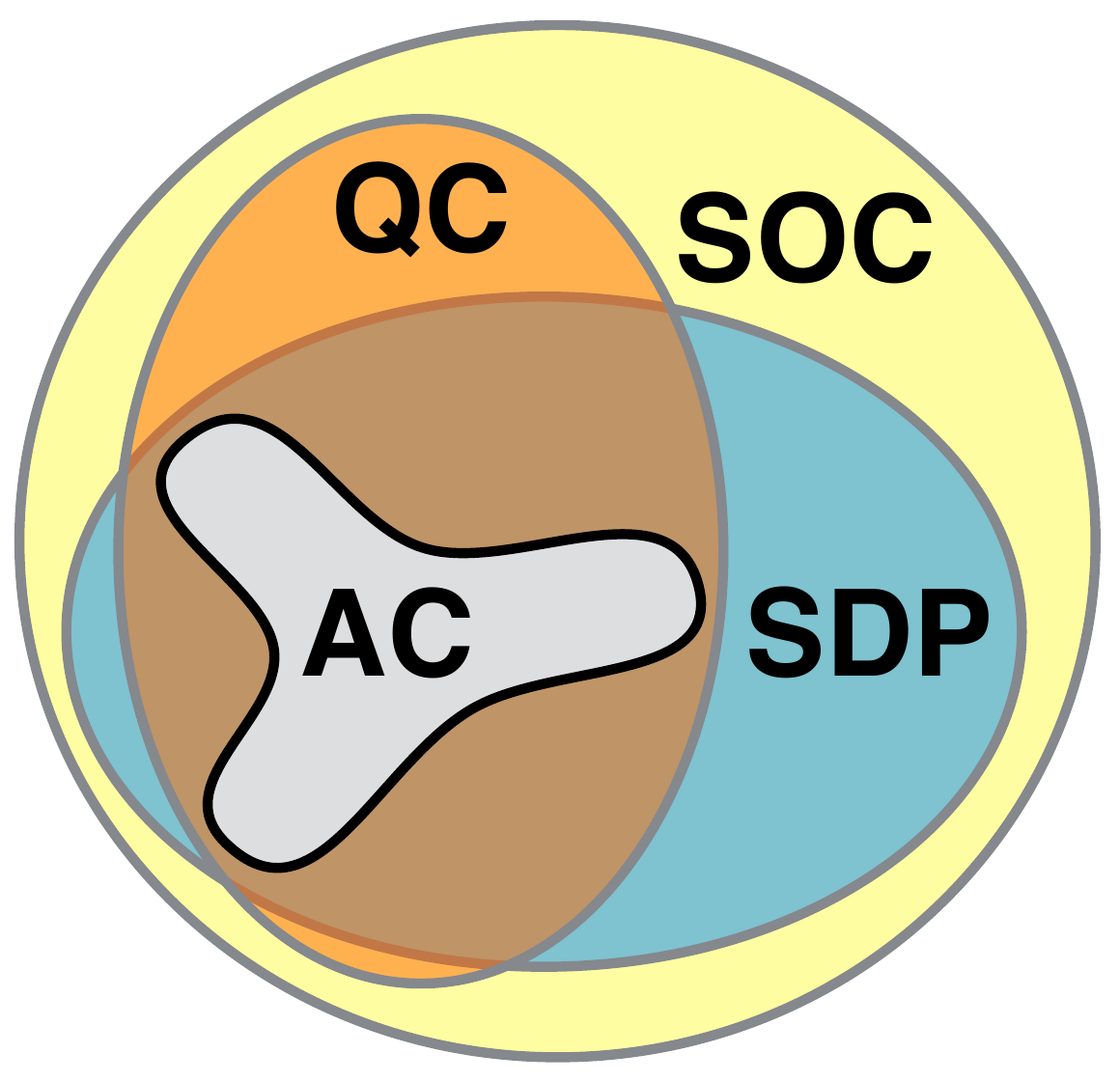}
\vspace{-0.2cm}
\caption{A Venn Diagram of the Solutions Sets for Various AC Power Flow Relaxations (set sizes in this illustration are not to scale).} 
\label{fig:relaxation_sets}
\end{figure}

\begin{theorem}
The following properties, illustrated in Figure \ref{fig:relaxation_sets}, hold:
\begin{enumerate}
\item $SDP \subseteq SOC$.
\item $SDP \neq QC$.
\item $QC \subseteq SOC$.
\end{enumerate}
\end{theorem}
\begin{proof}
  Properties (1) and (2) follows from \cite{6345272} and Section
  \ref{sec:example} respectively. For Property (3), observe that the
  set of constraints in Model \ref{model:ac_opf_w_qc} (W-QC) is a
  superset of those in Model \ref{model:ac_opf_w_soc}. The result
  follows from Corollary \ref{corollary:model}.
\end{proof}

\noindent
Observe that the additional constraints \eqref{qc_1}--\eqref{qc_3} in
the QC formulations are parameterized by the $\theta^\Delta$.  As
$\theta^\Delta$ grows larger, the QC model reduces to the SOC
model. Clearly, the strength of the QC relaxation is sensitive to this
input parameter, as illustrated in Section \ref{sec:example}.

%
%

\section{Computational Evaluation}
\label{sec:expar}

This section presents a computational evaluation of the relaxations and address the following questions:
\begin{enumerate}
\item How big are the optimality gaps in practice?
\item What are the runtime requirements of the relaxations?
\item How robust is the solving technology for the relaxations?
\end{enumerate}
The relaxations were compared on 105 state-of-the-art AC-OPF
transmission system test cases from the NESTA v0.4.0 archive
\cite{nesta}.  These test cases range from as few as 3 buses to as
many as 9000 and consist of 35 different networks under a typical
operating condition (TYP), a congested operating condition (API), and
a small angle difference condition (SAD).\footnote{Nine test cases
  based on the EIR Grid network were omitted from evaluation because
  the AC-OPF-W-SDP solver did not support inactive buses.}

\paragraph*{Experimental Setting}

All of the computations are conducted on \amdq. IPOPT 3.12
\cite{Ipopt} with linear solver ma27 \cite{hsl_lib}, as suggested by \cite{acopf_solvers}, was used as a
heuristic for finding locally optimal feasible solutions to the non-convex AC-OPF formulated in AMPL \cite{ampl}. 
The SDP relaxation was executed on
the state-of-the-art implementation \cite{opfBranchDecompImpl} which
uses a branch decomposition \cite{opfBranchDecomp} with a minor
extension to add constraint \eqref{w_9}.  The SDP solver SDPT3 4.0
\cite{Toh99sdpt3} was used with the modifications suggested in
\cite{opfBranchDecompImpl}.
The second-order cone models were formulated in AMPL 
and IPOPT was used to solve the models.
Numerical stability appears to be a significant challenge on the power
networks with more than 1000 buses \cite{BienstockCH14}. 
Note that IPOPT is single-threaded and does
not take advantage of the multiple cores available in the computation
servers.  This gives some computational advantage to the SDP solver,
which utilizes multiple cores.
%

\paragraph*{Challenging Test Cases}

We observe that 52 of the 105 test cases considered have an optimality
gap of less than 1.0\% with the SOC relaxation.  Such test cases are
not particularly useful for this study as the improvements of the SDP
and QC models are minor.  Hence, we focus our attention on the 53 test
cases where the SOC optimality gap is greater than 1.0\%. The results
are displayed in Table \ref{tbl:gaps_time}.

\begin{table*}[t!]
\footnotesize
\center
\caption{Quality and Runtime Results of AC Power Flow Relaxations}
\vspace{-0.2cm}
\begin{tabular}{|r||r||r|r|r|r||r|r|r|r|r|r|r|r|r|c|c|}
\hline
                 & \$/h & \multicolumn{4}{c||}{Optimality Gap (\%)} & \multicolumn{5}{c|}{Runtime (seconds)} \\
Test Case & AC & SDP & QC & SOC & CP & AC & SDP & QC & SOC & CP \\
\hline
\hline
\multicolumn{11}{|c|}{Typical Operating Conditions (TYP)} \\
\hline
 nesta\_case3\_lmbd & 5812.64 & 0.39 & 1.24 & 1.32 & 2.99 & 0.12 & 4.16 & 0.07 & 0.05 & 0.03 \\
\hline
 nesta\_case5\_pjm & 17551.89 & 5.22 & 14.54 & 14.54 & 15.62 & 0.04 & 5.36 & 0.09 & 0.03 & 0.05 \\
\hline
 nesta\_case30\_ieee & 204.97 & \bf 0.00  & 15.64 & 15.88 & 27.91 & 0.09 & 8.38 & 0.17 & 0.07 & 0.06 \\
\hline
 nesta\_case118\_ieee & 3718.64 & 0.06 & 1.72 & 2.07 & 7.87 & 0.41 & 12.62 & 0.87 & 0.43 & 0.05 \\
\hline
 nesta\_case162\_ieee\_dtc & 4230.23 & 1.08 & 4.00 & 4.03 & 15.44 & 0.61 & 35.20 & 1.48 & 0.31 & 0.04 \\
\hline
 nesta\_case300\_ieee & 16891.28 & 0.08 & 1.17 & 1.18 & n.a. & 0.80 & 29.69 & 2.83 & 0.65 & n.a. \\
\hline
 nesta\_case2224\_edin & 38127.69 & 1.22 & 6.03 & 6.09 & 8.45 & 11.42 & 690.16 & 65.59 & 45.99 & 0.33 \\
\hline
 nesta\_case2383wp\_mp & 1868511.78 & 0.37 & 1.04 & 1.05 & 5.35 & 12.41 & 1966.10 & 57.87 & 12.91 & 0.80 \\
\hline
 nesta\_case3012wp\_mp & 2600842.72 & --- & 1.00 & 1.02 & n.a. & 12.40 & 14588.79$^\dagger$ & 53.59 & 19.15 & n.a. \\
\hline
 nesta\_case9241\_pegase & 315913.26 & --- & 1.67 & --- & n.a. & 132.25 & --- & 3064.42 & --- & n.a. \\
\hline
\hline
\multicolumn{11}{|c|}{Congested Operating Conditions (API)} \\
\hline
 nesta\_case3\_lmbd\_\_api & 367.74 & 1.26 & 1.83 & 3.30 & 14.79 & 0.18 & 4.41 & 0.09 & 0.05 & 0.23 \\
\hline
 nesta\_case6\_ww\_\_api & 273.76 & 0.00$^\star$ & 13.14 & 13.33 & 17.17 & 0.34 & 13.19 & 0.07 & 0.06 & 0.03 \\
\hline
 nesta\_case14\_ieee\_\_api & 325.56 & \bf 0.00  & 1.34 & 1.34 & 8.89 & 0.19 & 5.64 & 0.11 & 0.08 & 0.94 \\
\hline
 nesta\_case24\_ieee\_rts\_\_api & 6421.37 & 1.45 & 13.77 & 20.70 & 24.12 & 0.14 & 7.50 & 0.26 & 0.09 & 0.04 \\
\hline
 nesta\_case30\_as\_\_api & 571.13 & 0.00 & 4.76 & 4.76 & 8.01 & 0.38 & 6.12 & 0.17 & 0.11 & 1.11 \\
\hline
 nesta\_case30\_fsr\_\_api & 372.14 & 11.06 & 45.97 & 45.97 & 48.80 & 0.25 & 7.25 & 0.19 & 0.09 & 0.92 \\
\hline
 nesta\_case30\_ieee\_\_api & 415.53 & 0.00 & 1.01 & 1.01 & 12.75 & 0.07 & 6.60 & 0.19 & 0.09 & 0.03 \\
\hline
 nesta\_case39\_epri\_\_api & 7466.25 & \bf 0.00  & 2.97 & 2.99 & 13.31 & 0.10 & 7.36 & 0.29 & 0.12 & 0.04 \\
\hline
 nesta\_case73\_ieee\_rts\_\_api & 20123.98 & 4.29 & 12.01 & 14.34 & 17.83 & 0.48 & 10.03 & 0.66 & 0.20 & 0.06 \\
\hline
 nesta\_case89\_pegase\_\_api & 4288.02 & 18.11 & 20.39 & 20.43 & 22.60 & 1.16 & 21.58 & 1.29 & 0.81 & 0.04 \\
\hline
 nesta\_case118\_ieee\_\_api & 10325.27 & 31.50 & 43.93 & 44.08 & 49.69 & 0.46 & 12.59 & 0.84 & 0.25 & 0.05 \\
\hline
 nesta\_case162\_ieee\_dtc\_\_api & 6111.68 & 0.85 & 1.33 & 1.34 & 19.39 & 0.50 & 36.85 & 1.53 & 0.39 & 0.05 \\
\hline
 nesta\_case189\_edin\_\_api & 1982.82 & 0.05 & 5.78 & 5.78 & n.a. & 1.07 & 16.10 & 1.14 & 0.33 & n.a. \\
\hline
 nesta\_case2224\_edin\_\_api & 46235.43 & 1.10 & 2.77 & 2.77 & 9.07 & 12.28 & 672.04 & 81.66 & 88.33 & 0.33 \\
\hline
 nesta\_case2383wp\_mp\_\_api & 23499.48 & 0.10 & 1.12 & 1.12 & 3.10 & 9.50 & 1421.39 & 28.37 & 10.25 & 0.34 \\
\hline
 nesta\_case2736sp\_mp\_\_api & 25437.70 & 0.07 & 1.32 & 1.33 & 3.89 & 9.21 & 2278.77 & 41.29 & 10.51 & 0.36 \\
\hline
 nesta\_case2737sop\_mp\_\_api & 21192.40 & 0.00 & 1.05 & 1.06 & 4.62 & 9.29 & 1887.22 & 30.94 & 9.91 & 0.32 \\
\hline
 nesta\_case2869\_pegase\_\_api & 96573.10 & 0.92$^\star$ & 1.49 & 1.49 & 5.16 & 21.03 & 1579.87 & 102.55 & 161.96 & 0.37 \\
\hline
 nesta\_case3120sp\_mp\_\_api & 22874.98 & --- & 3.02 & 3.03 & n.a. & 14.92 & 15018.93$^\dagger$ & 41.72 & 12.19 & n.a. \\
\hline
 nesta\_case9241\_pegase\_\_api & 241975.18 & --- & 2.45 & 2.59 & n.a. & 140.73 & --- & 3511.60 & 8387.11 & n.a. \\
\hline
\hline
\multicolumn{11}{|c|}{Small Angle Difference Conditions (SAD)} \\
\hline
 nesta\_case3\_lmbd\_\_sad & 5992.72 & 2.06 & 1.24$^\star$ & 4.28 & 5.90 & 0.19 & 4.39 & 0.10 & 0.05 & 0.03 \\
\hline
 nesta\_case4\_gs\_\_sad & 324.02 & 0.05 & 0.81 & 4.90 & 66.06 & 0.24 & 4.16 & 0.06 & 0.06 & 0.07 \\
\hline
 nesta\_case5\_pjm\_\_sad & 26423.32 & \bf 0.00  & 1.10 & 3.61 & 43.95 & 0.08 & 5.35 & 0.11 & 0.05 & 0.03 \\
\hline
 nesta\_case6\_c\_\_sad & 24.43 & 0.00 & 0.40 & 1.36 & 6.79 & 0.26 & 5.32 & 0.11 & 0.05 & 0.02 \\
\hline
 nesta\_case9\_wscc\_\_sad & 5590.09 & 0.00 & 0.41 & 1.50 & 6.69 & 0.14 & 4.18 & 0.19 & 0.05 & 0.03 \\
\hline
 nesta\_case24\_ieee\_rts\_\_sad & 79804.96 & 6.05 & 3.88 & 11.42 & 23.56 & 0.10 & 6.24 & 0.30 & 0.11 & 0.04 \\
\hline
 nesta\_case29\_edin\_\_sad & 46933.26 & 28.44 & 20.57 & 34.47 & 36.79 & 0.70 & 9.19 & 1.73 & 0.27 & 0.06 \\
\hline
 nesta\_case30\_as\_\_sad & 914.44 & 0.47 & 3.07 & 9.16 & 16.06 & 0.18 & 6.49 & 0.22 & 0.09 & 0.03 \\
\hline
 nesta\_case30\_ieee\_\_sad & 205.11 & \bf 0.00  & 3.96 & 5.84 & 27.96 & 0.12 & 7.49 & 0.18 & 0.09 & 0.03 \\
\hline
 nesta\_case73\_ieee\_rts\_\_sad & 235241.70 & 4.10 & 3.51 & 8.37 & 22.21 & 0.30 & 9.48 & 0.87 & 0.20 & 0.07 \\
\hline
 nesta\_case118\_ieee\_\_sad & 4324.17 & 7.57 & 8.32 & 12.89 & 20.77 & 0.56 & 14.14 & 0.98 & 0.31 & 0.06 \\
\hline
 nesta\_case162\_ieee\_dtc\_\_sad & 4369.19 & 3.65 & 6.91 & 7.08 & 18.13 & 0.81 & 39.71 & 1.70 & 0.36 & 0.05 \\
\hline
 nesta\_case189\_edin\_\_sad & 914.61 & 1.20$^\star$ & 2.22 & 2.25 & n.a. & 0.65 & 14.83 & 1.27 & 0.46 & n.a. \\
\hline
 nesta\_case300\_ieee\_\_sad & 16910.23 & 0.13 & 1.16 & 1.26 & n.a. & 1.01 & 29.63 & 2.81 & 0.76 & n.a. \\
\hline
 nesta\_case2224\_edin\_\_sad & 38385.14 & 1.22 & 5.57 & 6.18 & 9.06 & 11.53 & 691.53 & 50.34 & 65.68 & 0.33 \\
\hline
 nesta\_case2383wp\_mp\_\_sad & 1935308.12 & 1.30 & 2.97 & 4.00 & 8.62 & 16.25 & 1785.26 & 40.71 & 12.57 & 0.80 \\
\hline
 nesta\_case2736sp\_mp\_\_sad & 1337042.77 & 2.18$^\star$ & 2.01 & 2.34 & 4.56 & 13.22 & 1737.25 & 35.42 & 11.31 & 0.48 \\
\hline
 nesta\_case2737sop\_mp\_\_sad & 795429.36 & 2.24$^\star$ & 2.21 & 2.42 & 3.95 & 13.01 & 2153.37 & 32.05 & 9.69 & 0.39 \\
\hline
 nesta\_case2746wp\_mp\_\_sad & 1672150.46 & 2.41$^\star$ & 1.83 & 2.44 & 5.43 & 14.01 & 2840.32 & 35.66 & 13.32 & 0.56 \\
\hline
 nesta\_case2746wop\_mp\_\_sad & 1241955.30 & 2.71$^\star$ & 2.48 & 2.94 & 5.14 & 14.51 & 2306.18 & 32.41 & 23.22 & 0.42 \\
\hline
 nesta\_case3012wp\_mp\_\_sad & 2635451.29 & --- & 1.92 & 2.12 & n.a. & 15.79 & 13548.13$^\dagger$ & 46.59 & 28.41 & n.a. \\
\hline
 nesta\_case3120sp\_mp\_\_sad & 2203807.23 & --- & 2.56 & 2.79 & n.a. & 30.01 & 16804.55$^\dagger$ & 53.81 & 15.69 & n.a. \\
\hline
 nesta\_case9241\_pegase\_\_sad & 315932.06 & --- & 0.80 & 1.75 & n.a. & 80.30 & --- & 3531.62 & 33437.86 & n.a. \\
\hline
\end{tabular}\\
\vspace{0.1cm}
{\bf bold} - the relaxation provided a feasible AC power flow, $\star$ - solver reported numerical accuracy warnings,  ---,$\dagger$ - iteration or memory limit
\label{tbl:gaps_time}
\end{table*}

\subsection{The Quality of the Relaxations}

The first six columns of Table \ref{tbl:gaps_time} present the
optimality gaps for each of the relaxations on the 53 challenging
NESTA test cases.  Note that bold values indicate cases where 
the relaxation produced a solution to the non-convex AC power flow 
problem.  The table illustrates that this is a rare occurrence in the cases 
considered.


\paragraph*{The SDP Relaxation}

Overall, the SDP relaxation tends to be the tightest, often featuring
optimality gaps below 1.0\%.  In 5 of the 53 cases, the SDP relaxation
even produces a feasible AC power flow solution (as first observed in
\cite{5971792}).  However, with six notable cases where the gap is
above 5\%, it is clear that small gaps are not guaranteed. In some cases,
the optimality gap can be as large as 30\%. 

A significant issue with the SDP relaxation is the reliability of the
solving technology.  Even after applying the solver modifications
suggested in \cite{opfBranchDecompImpl}, the solver fails to converge
to a solution before hitting the default iteration limit on 8 of the
53 test cases shown, it reports numerical accuracy warnings on 4 of the test cases,
and ran out of memory on the 3 test cases with more 9000 nodes.

\paragraph*{The QC and SOC Relaxations}

As suggested by the theoretical study in Section \ref{sec:relations},
when the phase angle difference bounds are large, the QC
relaxation is quite similar to the SOC relaxation.  However, when the
phase angle difference bounds are tight (e.g., in the SAD cases), the
QC relaxation has significant benefits over the SOC relaxation.  On
average, the SDP relaxation dominates the QC and SOC relaxations.
However, there are several notable cases
(e.g. nesta\_case24\_ieee\_rts\_\_sad, nesta\_case29\_edin\_\_sad,
nesta\_case73\_ieee\_rts\_\_sad) where the QC relaxation dominates the
SDP relaxation.

\paragraph*{The Copper Plate (CP) Relaxation}
This relaxation indicates the cost of supplying power 
to the loads when there are no line losses or network constraints \cite{nfcp_report}, 
and is included in the table as a point of reference.
Note that this relaxation cannot be applied to networks containing 
lines with negative resistance or impedance, as indicated by ``n.a.''.

\subsection{The Performance of the Relaxations}

Detailed runtime results for the heuristic solution method and the
relaxations are presented in the last four columns of Table
\ref{tbl:gaps_time}.  The AC heuristic is fast, often taking less than
1 second on test cases with less than 1000 buses.  The SOC relaxation
most often has very similar performance to the AC heuristic.  The
additional constraints in the QC relaxation add a factor 2--5 on top
of the SOC relaxation.  In contrast to these other methods, the SDP
relaxation stands out, taking 10--100 times longer.  It is interesting
to observe, in the 5 cases where the SDP relaxation finds an
AC-feasible solution, the heuristic finds a solution of equal quality
in a fraction of the time.  Focusing on the test cases where the SDP
fails to converge, we observe that the failure occurs after several
minutes of computation, further emphasizing the reliability issue.

\section{Conclusion}
\label{sec:conclusion}

This paper compared the QC relaxation of the power flow equations with
the well-understood SDP and SOC relaxations both theoretically and
experimentally. Its two main contributions are as follows:

\begin{enumerate}
\item The QC relaxation is stronger than the SOC relaxation and
  neither dominates nor is dominated by the SDP relaxation.

\item Computational results on optimal power flow show that the QC
  relaxation may bring significant benefits in accuracy over the SOC
  relaxation, especially for tight bounds on phase angle differences,
  for a reasonable loss in efficiency. In addition, they show that,
  with existing solvers, the SOC and QC relaxations are significantly
  faster and more reliable than the SDP relaxation.
\end{enumerate}

\noindent
There are two natural frontiers for future work on these relaxations;
One is to utilize these relaxations in power system applications that
are modeled as mixed-integer nonlinear optimization problems, such as the
Optimal Transmission Switching, Unit Commitment, 
or Transmission Network Expansion Planning. 
Indeed, Mixed-Integer Quadratic Programming solvers
are already being used to extend these relaxations to richer power
system applications \cite{QCarchive, pscc_ots, 6407493, pscc_dsr,
  6308747}.  The other frontier is to develop novel methods for
closing the significant optimality gaps that remain on a variety of
test cases considered here.

\section*{Acknowledgements}

The authors would like to thank the four anonymous reviewers for their insightful suggestions for improving this work.
NICTA is funded by the Australian Government through the Department of
Communications and the Australian Research Council through the ICT
Centre of Excellence Program.

\bibliographystyle{IEEEtran}
\bibliography{../../power_models}

\begin{IEEEbiographynophoto}{Carleton Coffrin}
received a B.Sc. in Computer Science and a B.F.A. in Theatrical 
Design from the University of Connecticut, Storrs, CT and a
M.S. and Ph.D. from Brown University, Providence, RI.  He is currently a
staff researcher at National ICT Australia where he studies the application 
of optimization methods to problems in power systems.
\end{IEEEbiographynophoto}

\begin{IEEEbiographynophoto}{Hassan L. Hijazi}
received a Ph.D. in Computer Science from AIX-Marseille University while working at Orange Labs - France Telecom R\&D from 2007 to 2010. He then joined the Optimization Group at the Computer Science Laboratory of the Ecole Polytechnique-France where he stayed until late 2012. He is currently a senior research scientist at National ICT Australia and a senior lecturer at the Australian National University.  His main field of research is mixed-integer nonlinear optimization and applications in network-based problems, where he has given contributions both in theory and practice.
\end{IEEEbiographynophoto}

\begin{IEEEbiographynophoto}{Pascal Van Hentenryck}
received the undergraduate and Ph.D. degrees from the University of Namur, Namur, Belgium.
He currently leads the Optimization Research Group at National ICT Australia and holds the Vice-Chancellor Chair in data-intensive computing at the Australian National University, Canberra, Australia. 
Prior to that, he was a Professor with Brown University, Providence, RI, USA. His current research interests are in optimization with applications to disaster management, power systems, and transportation.
\end{IEEEbiographynophoto}


\appendix

\section*{Proof of Lemma \ref{lemma:properties}}
\label{appendix-proof}

\begin{proof} $\;$ \\
Property 1 -- the absolute square of power --
%
\begin{align}
& |S_{ij}|^2 = |\bm Y_{ij}|^2 \left( W_{ii}^2 - W_{ii} W_{ij} - W_{ii} W^*_{ij} + |W_{ij}|^2 \right) \nonumber
\end{align}
is derived using the following steps:
\begin{subequations}
\begin{align}
& S_{ij} = \bm Y^*_{ij} W_{ii} - \bm Y^*_{ij} W_{ij} \nonumber \\
& S_{ij} S^*_{ij} = (\bm Y^*_{ij} W_{ii} - \bm Y^*_{ij} W_{ij})(\bm Y_{ij} W^*_{ii} - \bm Y_{ij} W^*_{ij}) \nonumber \\
& |S_{ij}|^2 = |\bm Y_{ij}|^2 W_{ii}W^*_{ii} - |\bm Y_{ij}|^2 W^*_{ii} W_{ij} - |\bm Y_{ij}|^2 W_{ii} W^*_{ij} \nonumber \\ 
& + |\bm Y_{ij}|^2 W_{ij} W^*_{ij} \nonumber \\
& |S_{ij}|^2 = |\bm Y_{ij}|^2 \left( W_{ii}^2 - W_{ii} W_{ij} - W_{ii} W^*_{ij} + |W_{ij}|^2 \right). \nonumber
\end{align}
\end{subequations}
\\
\noindent
Property 2 -- the absolute square of the voltage product --
\begin{align}
& |W_{ij}|^2 = W_{ii}^2 - W^*_{ii} \bm Z^*_{ij} S_{ij} - W_{ii} \bm Z_{ij} S^*_{ij} + |\bm Z_{ij}|^2 |S_{ij}|^2 \nonumber
\end{align}
is derived using the following steps:
\begin{subequations}
\begin{align}
& S_{ij} = \bm Y^*_{ij} W_{ii} - \bm Y^*_{ij} W_{ij} \nonumber \\
& W_{ij} = W_{ii} - \bm Z^*_{ij} S_{ij} \nonumber \\
& W_{ij} W^*_{ij} = (W_{ii} - \bm Z^*_{ij} S_{ij})(W^*_{ii} - \bm Z_{ij} S^*_{ij}) \nonumber \\
& |W_{ij}|^2 = W_{ii}^2 - W^*_{ii} \bm Z^*_{ij} S_{ij} - W_{ii} \bm Z_{ij} S^*_{ij} + |\bm Z_{ij}|^2 |S_{ij}|^2. \nonumber
\end{align}
\end{subequations}
\\
\noindent
Property 3 -- the absolute square of current --
\begin{align}
& l_{ij} = |\bm Y_{ij}|^2 (W_{ii} + W_{jj} - W_{ij} - W^*_{ij})  \nonumber
\end{align}
is derived using the following steps:
\begin{subequations}
\begin{align}
&  I_{ij} = \bm Y_{ij} (V_i - V_j) \nonumber \\
&  I_{ij}I^*_{ij} = \bm Y_{ij} (V_i - V_j) \bm Y^*_{ij} (V^*_i - V^*_j) \nonumber \\
&  |I_{ij}|^2 = |\bm Y_{ij}|^2(V_iV^*_i - V_iV^*_j - V^*_iV_j + V_jV^*_j) \nonumber \\
&  l_{ij} = |\bm Y_{ij}|^2(W_{ii} - W_{ij} - W^*_{ij} + W_{jj}). \nonumber
%
\end{align}
\end{subequations}
\\
\noindent
Property 4 -- voltage drop --
\begin{align}
& W_{jj} = W_{ii} - \bm Z^*_{ij} S_{ij} - \bm Z_{ij} S^*_{ij} + |\bm Z_{ij}|^2 l_{ij} \nonumber
\end{align}
is derived using the following steps:
\begin{subequations}
\begin{align}
& W_{jj} = W_{jj} \nonumber \\
& W_{jj} = W_{ii} - W_{ii} + W_{ij}  - W_{ii} + W^*_{ij} \nonumber \\ 
& + W_{ii} - W_{ij} - W^*_{ij} + W_{jj} \nonumber \\
& W_{jj} = W_{ii} - W_{ii} + W_{ij}  - W_{ii} + W^*_{ij} + |\bm Z_{ij}|^2 l_{ij} \nonumber \\
& W_{jj} = W_{ii} - \bm Z^*_{ij} S_{ij} - \bm Z_{ij} S^*_{ij} + |\bm Z_{ij}|^2 l_{ij}. \nonumber
\end{align}
\end{subequations}
Note that property 3 is used in the second step.
\end{proof}

\section*{Extensions for Transmission System Test Cases}
\label{sec:extensions}

In the interest of clarity, properties of AC Power Flow, and their
relaxations, are most often presented on the purest version of the AC
power flow equations.  However, transmission system test cases include
additional parameters such as bus shunts, line charging, and
transformers, which complicate the AC power flow equations
significantly.  Model \ref{model:ac_opf_w_ext} presents the AC Optimal
Power Flow problem (similar to Model \ref{model:ac_opf_w}) with these
extensions.  In the rest of this section, we show that the results of
Section \ref{sec:qc_alt} continue to hold in this extended power flow
model.

\begin{model}[t]
\caption{ AC-OPF-W with Extensions}
\label{model:ac_opf_w_ext}
\begin{subequations}
\vspace{-0.2cm}
\begin{align}
\mbox{\bf variables: } & S^g_i (\forall i\in N), \; W_{ij} (\forall i,j \in N) \nonumber \\
%
\mbox{\bf minimize: } & \eqref{w_obj} \\
\mbox{\bf subject to: } &  \mbox{\eqref{w_2}--\eqref{w_4}, \eqref{w_8}--\eqref{w_9}} \nonumber \\
%
%
& \hspace{-1.5cm} S^g_i - {\bm S^d_i} - \bm Y^s_{i} W_{ii} = \sum_{\substack{(i,j)\in E \cup E^R}} S_{ij} \;\; \forall i\in N \label{w2_1} \\ 
& \hspace{-1.5cm} S_{ij} = \left( \bm Y^*_{ij} - \bm i\frac{\bm {b^c}_{ij}}{2} \right) \frac{W_{ii}}{|\bm{T}_{ij}|^2} - \bm Y^*_{ij} \frac{W_{ij}}{\bm{T}^*_{ij}} \;\; (i,j)\in E \label{w2_2}\\
& \hspace{-1.5cm} S_{ji} = \left( \bm Y^*_{ij} - \bm i\frac{\bm {b^c}_{ij}}{2} \right) W_{jj} - \bm Y^*_{ij} \frac{W^*_{ij}}{\bm{T}_{ij}} \;\; (i,j)\in E \label{w2_3}
\end{align}
\end{subequations}
\end{model}

\subsubsection*{The Two SOC Formulations}

In this extended power flow formulation, the second-order cone
constraint based on the absolute square of the voltage product
\cite{Jabr06} remains the same, i.e.,
\begin{align}
& |W_{ij}|^2 \leq W_{ii}W_{jj}  
\end{align}
However, the constraint based on the absolute square of the power flow
\cite{6102366} is updated to include the transformer tap ratio as
follows:
\begin{align}
& |S_{ij}|^2 \leq \frac{W_{ii}}{|\bm{T}_{ij}|^2}l_{ij}
\end{align}
and the power loss constraint \eqref{qc_4} is updated to
\begin{align}
& S_{ij} + S_{ji} = \bm Z_{ij} \left( l_{ij} + \left( \frac{\bm {b^c}_{ij}}{2} \right)^{\!2} \! \frac{W_{ii}}{|\bm T_{ij}|^2} + \bm {b^c}_{ij} \Im(S_{ij}) \right) \nonumber \\ 
& - \bm i \frac{\bm {b^c}_{ij}}{2} \left( \frac{W_{ii}}{|\bm T_{ij}|^2} + W_{jj} \right) \label{eq:loss_ext} 
\end{align}

\noindent
We now show that, when the power flow equations
\eqref{w2_2}--\eqref{w2_3} are present in the model, these two
versions of the second-order cone constraints are also equivalent,
i.e.,
\begin{equation}
\label{soc_w_e} \tag{W-E-SOC}
\begin{array}{ll}
& S_{ij} = \left( \bm Y^*_{ij} - \bm i\frac{\bm {b^c}_{ij}}{2} \right) \frac{W_{ii}}{|\bm{T}_{ij}|^2} - \bm Y^*_{ij} \frac{W_{ij}}{\bm{T}^*_{ij}} \;\; (i,j) \in E \\
& S_{ji} = \left( \bm Y^*_{ij} - \bm i\frac{\bm {b^c}_{ij}}{2} \right) W_{jj} - \bm Y^*_{ij} \frac{W^*_{ij}}{\bm{T}_{ij}} \;\; (i,j)\in E \\
& |W_{ij}|^2 \leq W_{ii}W_{jj} \;\; (i,j)\in E  \\
\end{array}
\end{equation}
is equivalent to
\begin{equation}
\label{soc_c_e} \tag{C-E-SOC}
\begin{array}{ll}
& S_{ij} = \left( \bm Y^*_{ij} - \bm i\frac{\bm {b^c}_{ij}}{2} \right) \frac{W_{ii}}{|\bm{T}_{ij}|^2} - \bm Y^*_{ij} \frac{W_{ij}}{\bm{T}^*_{ij}} \;\; (i,j) \in E \\
& S_{ji} = \left( \bm Y^*_{ij} - \bm i\frac{\bm {b^c}_{ij}}{2} \right) W_{jj} - \bm Y^*_{ij} \frac{W^*_{ij}}{\bm{T}_{ij}} \;\; (i,j)\in E \\
& S_{ij} + S_{ji} = \bm Z_{ij} \left( l_{ij} + \left( \frac{\bm {b^c}_{ij}}{2} \right)^{\!2} \! \frac{W_{ii}}{|\bm T_{ij}|^2} + \bm {b^c}_{ij} \Im(S_{ij}) \right) \\ 
& - \bm i \frac{\bm {b^c}_{ij}}{2} \left( \frac{W_{ii}}{|\bm T_{ij}|^2} + W_{jj} \right) \;\; (i,j) \in E \\
& |S_{ij}|^2 \leq \frac{W_{ii}}{|\bm T_{ij}|^2} l_{ij} \;\; (i,j)\in E.
\end{array}
\end{equation}

We begin by redeveloping the properties of Lemma \ref{lemma:properties} in the extended model.

\subsubsection*{The Equalities}

As both models contain constraints \eqref{w2_2}--\eqref{w2_3}, the
properties arising from these equations can be transferred between
both models.  

\begin{proof} $\;$ \\
\noindent
Property 1 -- the absolute square of power --
%
\begin{align}
& |S_{ij}|^2 = |\bm Y_{ij}|^2 \left( 
\frac{W_{ii}^2}{|\bm T_{ij}|^4} 
- \frac{W_{ii}}{|\bm T_{ij}|^2} \frac{W_{ij}}{\bm T^*_{ij}}
- \frac{W_{ii}}{|\bm T_{ij}|^2} \frac{W^*_{ij}}{\bm T_{ij}}
+ \frac{|W_{ij}|^2}{|\bm T_{ij}|^2} 
\right) \nonumber \\ 
& - \left( \frac{\bm {b^c}_{ij}}{2} \right)^{\!2} \! \frac{W_{ii}^2}{|\bm T_{ij}|^4} 
- \bm {b^c}_{ij} \frac{W_{ii}}{|\bm T_{ij}|^2} \Im(S_{ij})
\end{align}
The derivation follows similarly to the one presented earlier and the details are left to the reader.  The only delicate point is to observe that three separate terms in the initial expansion can be collected into $\Im(S_{ij})$. \\

%
\noindent
Property 2 -- the absolute square of the voltage product --
\begin{align}
& |W_{ij}|^2 = 
(1- \bm {b^c}_{ij} \Im(\bm Z_{ij})) \frac{W_{ii}^2}{|\bm T_{ij}|^2} 
- W^*_{ii} \bm Z^*_{ij} S_{ij} 
- W_{ii} \bm Z_{ij} S^*_{ij} \nonumber \\
& + |\bm Z_{ij}|^2 \left( 
   |\bm T_{ij}|^2 |S_{ij}|^2 
   + \left( \frac{\bm {b^c}_{ij}}{2} \right)^{\!2} \! \frac{W_{ii}^2}{|\bm T_{ij}|^2} 
   + W_{ii} \bm {b^c}_{ij} \Im(S_{ij})
\right)
\end{align}
The derivation follows similarly to the one presented earlier and the details are left to the reader. \\

\noindent
Property 3 -- the absolute square of current --
\begin{align}
& l_{ij} = |\bm Y_{ij}|^2 \left( \frac{W_{ii}}{|\bm T_{ij}|^2} - \frac{W_{ij}}{\bm T^*_{ij}} - \frac{W^*_{ij}}{\bm T_{ij}} + W_{jj}  \right) 
   \nonumber \\
   & - \left( \frac{\bm {b^c}_{ij}}{2} \right)^{\!2} \! \frac{W_{ii}}{|\bm T_{ij}|^2} 
   - \bm {b^c}_{ij} \Im(S_{ij})
\end{align}
After observing that the extension of Ohm's Law in this model is given by
\begin{align}
& I_{ij} = \left( \bm Y_{ij} + \bm i \frac{\bm {b^c}_{ij}}{2} \right) \frac{V_i}{\bm T_{ij}} - \bm Y_{ij}V_j  \;\; (i,j)\in E,
\end{align}
the derivation follows similarly to the one presented earlier and the details are left to the reader. \\

\noindent
Property 4 -- voltage drop --
\begin{align}
& W_{jj} = 
 (1- \bm {b^c}_{ij} \Im(\bm Z_{ij})) \frac{W_{ii}}{|\bm T_{ij}|^2} 
 -  \bm Z^*_{ij} S_{ij} - \bm Z_{ij} S^*_{ij} 
 \nonumber \\ 
& + |\bm Z_{ij}|^2 \left( 
   l_{ij} 
   + \left( \frac{\bm {b^c}_{ij}}{2} \right)^{\!2} \! \frac{W_{ii}}{|\bm T_{ij}|^2} 
   + \bm {b^c}_{ij} \Im(S_{ij})
\right)
\end{align}
The proof follows similarly to the one presented earlier and the
details are left to the reader.
\end{proof}

\noindent
With these core properties updated, we are now ready to extend the proof from Section \ref{sec:qc_alt}. 

\begin{theorem}
\label{theorem:main}
\eqref{soc_c_e} is equivalent to \eqref{soc_w_e}.
\end{theorem}
\begin{proof}

The proof follows the one presented in Section \ref{sec:qc_alt}.

\paragraph{\ref{soc_w_e} $\Rightarrow$ \ref{soc_c_e}} Every solution to \eqref{soc_w_e} is a solution to \eqref{soc_c_e}.
Given a solution to \eqref{soc_w}, by equality (3), we assign $l_{ij}$ as follows:
\begin{align}
& l_{ij} = |\bm Y_{ij}|^2 \left( \frac{W_{ii}}{|\bm T_{ij}|^2} - \frac{W_{ij}}{\bm T^*_{ij}} - \frac{ W^*_{ij}}{\bm T_{ij}} +  W_{jj}  \right) 
   \nonumber \\ 
   & - \left( \frac{\bm {b^c}_{ij}}{2} \right)^{\!2} \! \frac{ W_{ii}}{|\bm T_{ij}|^2} 
   - \bm {b^c}_{ij} \Im( S_{ij}) \;\; (i,j) \in E 
\end{align}
This assignment satisfies the power loss constraint \eqref{eq:loss_ext} by definition of the power.
It remains to show that second-order cone constraint in \eqref{soc_c_e} is satisfied. Using equalities (1) and (3), we obtain
\begin{subequations}
\begin{align}
& |S_{ij}|^2 = |\bm Y_{ij}|^2 \left( 
\frac{ W_{ii}^2}{|\bm T_{ij}|^4} 
- \frac{ W_{ii}}{|\bm T_{ij}|^2} \frac{ W_{ij}}{\bm T^*_{ij}}
- \frac{ W_{ii}}{|\bm T_{ij}|^2} \frac{ W^*_{ij}}{\bm T_{ij}}
+ \frac{| W_{ij}|^2}{|\bm T_{ij}|^2} 
\right) \nonumber \\
& - \left( \frac{\bm {b^c}_{ij}}{2} \right)^{\!2} \! \frac{ W_{ii}^2}{|\bm T_{ij}|^4} 
- \bm {b^c}_{ij} \frac{ W_{ii}}{|\bm T_{ij}|^2} \Im( S_{ij}) \\
& | S_{ij}|^2 \leq |\bm Y_{ij}|^2 \left( 
\frac{ W_{ii}^2}{|\bm T_{ij}|^4} 
\!-\! \frac{ W_{ii}}{|\bm T_{ij}|^2} \frac{ W_{ij}}{\bm T^*_{ij}}
\!-\! \frac{ W_{ii}}{|\bm T_{ij}|^2} \frac{ W^*_{ij}}{\bm T_{ij}}
\!+\! \frac{ W_{ii}  W_{jj}}{|\bm T_{ij}|^2} 
\right) \nonumber \\
& - \left( \frac{\bm {b^c}_{ij}}{2} \right)^{\!2} \! \frac{ W_{ii}^2}{|\bm T_{ij}|^4} 
- \bm {b^c}_{ij} \frac{ W_{ii}}{|\bm T_{ij}|^2} \Im( S_{ij}) \\
& | S_{ij}|^2 \leq \frac{ W_{ii}}{|\bm T_{ij}|^2} \left( |\bm Y_{ij}|^2 \left( 
\frac{ W_{ii}}{|\bm T_{ij}|^2} 
- \frac{ W_{ij}}{\bm T^*_{ij}}
- \frac{ W^*_{ij}}{\bm T_{ij}}
+  W_{jj}
\right) \right) \nonumber \\
& - \frac{ W_{ii}}{|\bm T_{ij}|^2} \left( \left( \frac{\bm {b^c}_{ij}}{2} \right)^{\!2} \! \frac{ W_{ii}}{|\bm T_{ij}|^2} 
- \bm {b^c}_{ij} \Im( S_{ij}) \right) \\
& | S_{ij}|^2 \leq \frac{ W_{ii}}{|\bm T_{ij}|^2}   l_{ij}.
\end{align}
\end{subequations}

\paragraph{\ref{soc_c_e} $\Rightarrow$ \ref{soc_w_e}} Every solution to \eqref{soc_c_e} is a solution to \eqref{soc_w_e}. 
We show that the values of $W_{ij}$ in \eqref{soc_c_e} satisfy the
second-order cone constraint in \eqref{soc_w_e}. Using equalities (2) and (4) and the fact
that $W_{ii} = W^*_{ii}$ since $W_{ii}$ is a real number, we have
\begin{subequations}
\begin{align}
& | W_{ij}|^2 = 
(1- \bm {b^c}_{ij} \Im(\bm Z_{ij})) \frac{ W_{ii}^2}{|\bm T_{ij}|^2} 
-  W^*_{ii} \bm Z^*_{ij}  S_{ij} 
-  W_{ii} \bm Z_{ij}  S^*_{ij} 
\nonumber \\
& + |\bm Z_{ij}|^2 \left( 
   |\bm T_{ij}|^2 | S_{ij}|^2 
   + \left( \frac{\bm {b^c}_{ij}}{2} \right)^{\!2} \! \frac{ W_{ii}^2}{|\bm T_{ij}|^2} 
   +  W_{ii} \bm {b^c}_{ij} \Im( S_{ij})
\right) \\
& | W_{ij}|^2 \leq 
(1- \bm {b^c}_{ij} \Im(\bm Z_{ij})) \frac{ W_{ii}^2}{|\bm T_{ij}|^2} 
-  W^*_{ii} \bm Z^*_{ij}  S_{ij} 
-  W_{ii} \bm Z_{ij}  S^*_{ij} 
\nonumber \\
& + |\bm Z_{ij}|^2 \left( 
     W_{ii}  l_{ij} 
   + \left( \frac{\bm {b^c}_{ij}}{2} \right)^{\!2} \! \frac{ W_{ii}^2}{|\bm T_{ij}|^2} 
   +  W_{ii} \bm {b^c}_{ij} \Im( S_{ij})
\right) \\
& | W_{ij}|^2 \leq  W_{ii} \left(
(1- \bm {b^c}_{ij} \Im(\bm Z_{ij})) \frac{ W_{ii}}{|\bm T_{ij}|^2} 
-  \bm Z^*_{ij}  S_{ij} 
-  \bm Z_{ij}  S^*_{ij}
\right) \nonumber \\ 
& +  W_{ii} \left( |\bm Z_{ij}|^2 \left( 
     l_{ij} 
   + \left( \frac{\bm {b^c}_{ij}}{2} \right)^{\!2} \! \frac{ W_{ii}}{|\bm T_{ij}|^2} 
   + \bm {b^c}_{ij} \Im( S_{ij})
\right) 
\right) \\
& | W_{ij}|^2 \leq  W_{ii}  W_{jj} 
%
\end{align}
\end{subequations}
and the result follows.
\end{proof}

\begin{corollary}
Model \ref{model:ac_opf_c_qc} is equivalent to Model \ref{model:ac_opf_w_qc} in the extended AC Power Flow formulation from Model \ref{model:ac_opf_w_ext}.
\end{corollary}

\begin{figure}[t!]
\centering
\includegraphics[width=9cm]{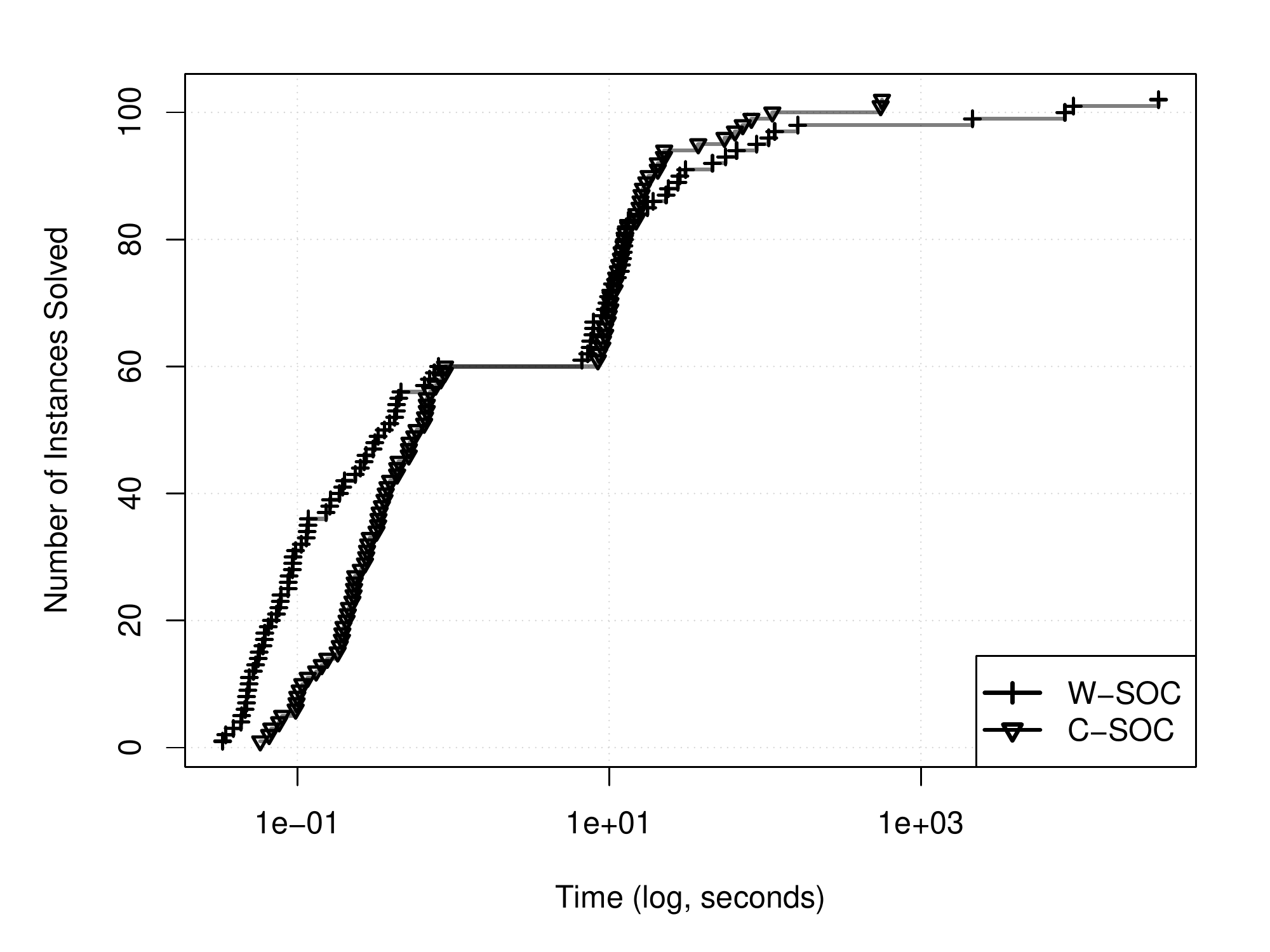}
\vspace{-0.7cm}
\caption{Runtime Profiles for the Two SOC Relaxations.}
\label{fig:soc_comp_soc}
\end{figure}

\begin{figure}[t!]
\centering
\includegraphics[width=9cm]{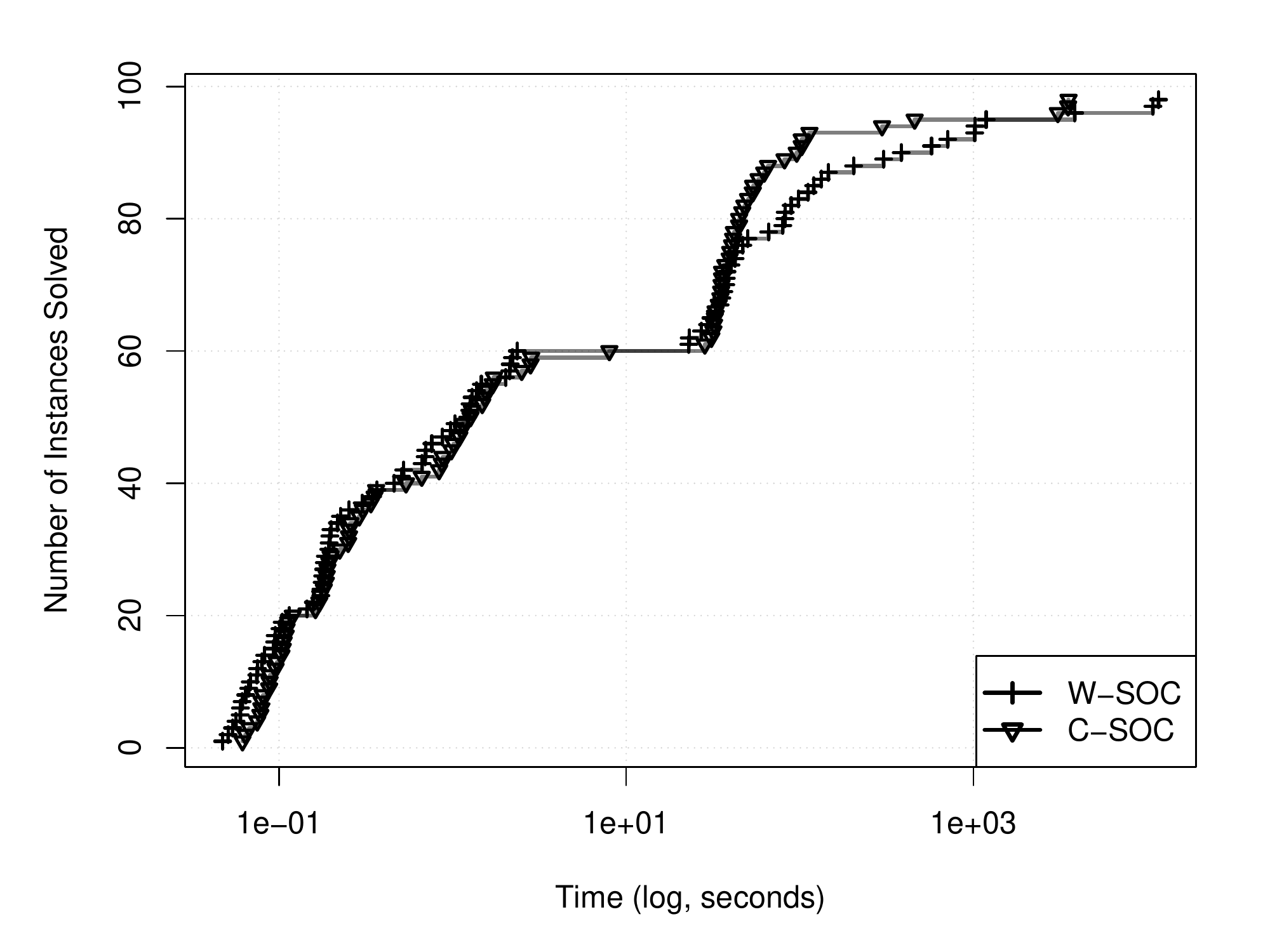}
\vspace{-0.7cm}
\caption{Runtime Profiles for the Two QC Relaxations.} 
\label{fig:soc_comp_qc}
\end{figure}

\section*{Comparison of SOC formulations}
\label{sec:soc_comp}

Section \ref{sec:qc_alt} proposed two equivalent formulations of the
second-order cone constraints for power flow relaxations.  Although
both formulations define the same convex set, it is unclear if they
have the same performance characteristics.  For example, the
current-based constraint \eqref{soc_c} has more constraints and more
variables than the voltage-product constraint \eqref{soc_w}.  All
other aspects being equal, one would expect \eqref{soc_c} to be slower
than \eqref{soc_w}.  This section investigates the performance
implications of these two formulations on both the QC and SOC power
flow relaxations. Four power flow relaxations are considered, W-SOC
(Model \ref{model:ac_opf_w_soc}), C-SOC (Model
\ref{model:ac_opf_w_soc} with \eqref{soc_c}), W-QC (Model
\ref{model:ac_opf_c_qc} with \eqref{soc_w}), and C-QC (Model
\ref{model:ac_opf_c_qc}).  To test the performance of these
relaxations, each model is evaluated on 105 state-of-the-art AC-OPF
transmission system test cases from the NESTA v0.4.0 archive
\cite{nesta}. Figure \ref{fig:soc_comp_soc} compares the two variants of the 
SOC relaxation and Figure \ref{fig:soc_comp_qc} compares two variants on the QC relaxation.


Both figures indicate that the two formulations are very similar for small test cases but, on the larger test cases (i.e., with more than 1000 buses), the C-QC formulation has a faster convergence rate, in IPOPT.  This suggests that, despite its increased size, the C-QC formulation originally
presented in \cite{QCarchive} is preferable from a performance standpoint and that the C-SOC formulation may be preferable on very large networks (e.g. above 9000 buses).


\end{document}